\documentclass{jgaa-art}

\newif\ifarxiv
\arxivtrue

\ifarxiv

\makeatletter
\renewcommand{\Issue}[5]{}
\renewcommand{\maketitle}{%
  \hypersetup{%
    pdftitle=\the\titlepdf@toks,
    pdfauthor=\the\authpdf@toks
  }%
  \pagestyle{plain}%
  \thispagestyle{plain}%
  \printtitle
  \begin{center}
    \begin{spacing}{1.1}
      \printauthors
    \end{spacing}
    \vspace{2mm}
    \printaffiliations
  \end{center}
  \begingroup
    \renewcommand*{\@makefnmark}{}%
    \footnotetext{%
      \ifnum\value{ack@count}>0
        \printack\par\smallskip
      \fi
      \printmail
    }%
  \endgroup
}
\makeatother

\else
    \usepackage{lineno}
    \linenumbers
\fi

\usepackage{graphicx}
\usepackage{amsmath}
\usepackage{amssymb}
\usepackage{todonotes}
\usepackage{subcaption}
\usepackage{zref-clever}
\zcsetup{cap, nameinlink=false}

\newtheorem{lemma}[theorem]{Lemma}
\newtheorem{observation}[theorem]{Observation}
\AddToHook{env/lemma/begin}{\zcsetup{countertype={theorem=lemma}}}
\AddToHook{env/observation/begin}{\zcsetup{countertype={theorem=observation}}}
\zcRefTypeSetup{observation}{
    Name-sg=Observation, name-sg=observation,
    Name-pl=Observations, name-pl=observations
}

\newcommand{\cross}{\chi}
\newcommand{\draw}{\mathrm{draw}}
\newcommand{\drawfan}{\mathrm{draw_{fan}}}
\newcommand{\OneSidedTableConstantFactor}{4100}

\begin{document}
\doi{} %
\Issue{0}{0}{0}{0}{0} %
\HeadingAuthor{Yasuaki Kobayashi and Yuto Okada} %
\HeadingTitle{2-Layer Fan-Planarity in Polynomial Time} %
\title{2-Layer Fan-Planarity in Polynomial Time} %
\Ack{A preliminary version of this paper appeared in the proceedings of GD 2026~\cite{KobayashiO026}. Yasuaki Kobayashi was supported by JSPS KAKENHI Grant Numbers JP23K28034, JP24H00686, JP24H00697, JP26K02980. Yuto Okada was supported by JST SPRING Grant Number JPMJSP2125 and JSPS KAKENHI Grant Number JP26KJ1299.} %

\authorOrcid[hokudai]{Yasuaki Kobayashi}{koba@ist.hokudai.ac.jp}{0000-0003-3244-6915} %
\authorOrcid[nagoya]{Yuto Okada}{research@yutookada.com}{0000-0002-1156-0383} %

\affiliation[hokudai]{Hokkaido University, Sapporo, Japan} %
\affiliation[nagoya]{Nagoya University, Nagoya, Japan, JSPS Research Fellow} %

\submitted{}%
\accepted{}%
\published{}%
\type{}%
\editor{}%

\setcounter{page}{1}
\maketitle
\thispagestyle{plain}

\begin{abstract}
    In this paper, we give a polynomial-time algorithm for deciding whether a bipartite graph admits a 2-layer fan-planar drawing. This resolves an open problem posed in several papers since 2015.
    We also obtain a linear-time algorithm for the one-sided case.
\end{abstract}

\section{Introduction}

A \emph{2-layer drawing} of a bipartite graph $G = (X \cup Y, E)$ is a drawing in the plane such that the vertices in $X$ and $Y$ are drawn on two horizontal lines, respectively, and each edge is drawn as a straight line segment between these two parallel lines.
This is one of the most natural drawing models for bipartite graphs and is also considered to be important in the famous Sugiyama's framework~\cite{SugiyamaTT81} of multi-layered graph drawing.
There are numerous studies regarding algorithmic aspects~\cite{AngeliniLBFP21,FominG00Z26,KobayashiO025,KobayashiT16} and graph-theoretic properties~\cite{AngeliniLFS24,Okada26,Wood23} for 2-layer drawings.

This paper focuses on a particular variant of 2-layer drawings called 2-layer \emph{fan-planar} drawings.
A 2-layer drawing is said to be \emph{fan-planar} if for every edge $e$, the edges crossing $e$ have a common end vertex; these edges form a ``fan''.
A bipartite graph admitting a 2-layer fan-planar drawing is called a \emph{2-layer fan-planar graph}.

The concept of (general) fan-planar drawings was introduced by Kaufmann and Ueckerdt~\cite{KaufmannU22}, which is a generalization of 1-planar drawings and a specialization of 3-quasiplanar drawings~\cite{KaufmannU22}.
The problem of determining whether an input graph admits a fan-planar drawing is known to be NP-complete~\cite{BinucciGDMPST15}.
The authors in~\cite{BinucciCDGKKMT17} studied 2-layer fan-planar drawings of bipartite graphs.
They gave a characterization of bipartite graphs having 2-layer fan-planar drawings and a linear-time algorithm for testing 2-layer fan-planarity for biconnected bipartite graphs~\cite{BinucciCDGKKMT17}.
Later, by exploiting their characterization, a linear-time algorithm for trees was given~\cite{BiedlC0MNR20}.
It is not clear whether this characterization would lead to a polynomial-time recognition algorithm for 2-layer fan-planar graphs.
In fact, a polynomial-time recognition algorithm for 2-layer fan-planar graphs was explicitly mentioned as an open problem in several papers~\cite{Bekos020,BinucciCDGKKMT17,BinucciGDMPST15,DidimoLM19}.

In this paper, we resolve this open question affirmatively by giving a polynomial-time algorithm for recognizing 2-layer fan-planar graphs.

\begin{theorem}\label{thm:main}
    There is a polynomial-time algorithm that, given a bipartite graph $G$, decides whether $G$ admits a 2-layer fan-planar drawing.
\end{theorem}

Our algorithm employs a different idea from those used in existing algorithms~\cite{BiedlC0MNR20,BinucciCDGKKMT17}.
For an integer $k\ge 0$, a 2-layer drawing of $G$ is said to be \emph{$k$-planar} if every edge involves at most $k$ crossings in the drawing.
Our polynomial-time algorithm exploits a recent polynomial-time algorithm for recognizing 2-layer $k$-planar graphs with fixed~$k$ \cite{KobayashiO025}.
It is not hard to see that, for every fixed $k$, there exists a 2-layer fan-planar drawing that is not $k$-planar.
To apply this algorithm for our purpose, we give some degree reduction rules, which enable us to reduce our problem to that on bounded-degree graphs~(\zcref{lem:degree-bound}).
Since every 2-layer fan-planar graph $G$ with maximum degree~$d$ has a 2-layer $d$-planar drawing that is fan-planar (\zcref{obs:k-planar}), by making a slight modification to the algorithm of \cite{KobayashiO025}, we can determine whether $G$ admits a 2-layer fan-planar drawing in polynomial time as well.

Moreover, in \zcref{sec:one-sided} we demonstrate that the same idea can also be applied to the \emph{one-sided} case, where one vertex set of the bipartition, say $X$, has its prescribed linear order in the input.
\begin{theorem}\label{thm:one-sided}
    There is a linear-time algorithm that, given a bipartite graph $G$ and a linear order $\sigma_X$ on $X$, decides whether $G$ admits a 2-layer fan-planar drawing that respects $\sigma_X$.
\end{theorem}
Odak, Rollin and Scheele~\cite{OdakRS026} recently studied 2-layer $k$-matching-planarity, which generalizes our problems as it is equivalent to 2-layer fan-planarity when $k=1$.
They gave an $O(n^3 + |X| 2^{2^{12k}})$ time algorithm for the one-sided case, one-sided $k$-matching-planarity, where $n = |V(G)|$.
Although this already shows the polynomial-time solvability, by focusing on the special case $k=1$, we obtain a linear-time algorithm, improving the hidden constant factor significantly at the same time.

\paragraph*{Notation and Terminology.}

Let $G = (X \cup Y, E)$ be a bipartite graph with two independent sets $X$ and $Y$.
For $U \subseteq X \cup Y$, the open (resp.~closed) neighborhood of $U$ is denoted by $N(U)$ (resp.~$N[U]$) and the set of edges between $U$ and $(X \cup Y) \setminus U$ is denoted by $\delta(U)$.
When $U$ is a singleton with $U = \{u\}$, we may simply write $N(u)$, $N[u]$, and $\delta(u)$ instead.
The subgraph of $G$ induced by a vertex set $U \subseteq X \cup Y$ is denoted by $G[U]$.
For a vertex $v$ and vertex set $U$ in $G$, we use $G - v$ and $G - U$ to denote $G[(X \cup Y) \setminus \{v\}]$ and $G[(X \cup Y) \setminus U]$, respectively.

A \emph{2-layer drawing} $\mathcal D$ of $G$ is a pair of linear orderings $(\sigma_X, \sigma_Y)$, where $\sigma_X$ and $\sigma_Y$ are defined on $X$ and $Y$, respectively.
We may interchangeably use a pair of binary relations on $\sigma_X$ and $\sigma_Y$ to represent a 2-layer drawing of $G$: For distinct $u, v \in X$ (resp.~distinct $u,v \in Y$), $u$ precedes $v$ on $\sigma_X$ (resp.~$\sigma_Y$) if and only if $u <_X v$ (resp.~$u <_Y v$).
We also consider a 2-layer drawing $\mathcal D$, where a linear order $\sigma_X$ of $X$ is given in advance.
In such cases, we say that $\mathcal D$ \emph{respects} $\sigma_X$.
A \emph{crossing} in $\mathcal D$ is a pair of edges $\{x, y\}$ and $\{x', y'\}$ such that $x <_X x'$ and $y' <_Y y$.
For $x, x' \in X$ with $x <_X x'$, we say that $x$ is \emph{to the left of} $x'$ ($x'$ is \emph{to the right of} $x$) in $\mathcal D$.
We also say $y$ is to the left (resp.~right) of $y'$ for $y, y' \in Y$ with $y <_Y y'$ (resp.~$y' <_Y y$).
The set of edges crossing an edge $e \in E$ is denoted by $\cross_{\mathcal D}(e)$.
We say that $\mathcal D$ is \emph{fan-planar} if for each edge in $G$, all edges in $\cross_{\mathcal D}(e)$ have a common end vertex.
Note that if $\mathcal D$ is not fan-planar, there are three edges $e \in E$ and $f, f' \in \cross_{\mathcal D}(e)$ that form a matching in $G$.
Such an unordered triple $\{e, f, f'\}$ is called a \emph{violating triple} in $\mathcal D$.

\section{2-Layer Fan-Planarity}

\subsection{Degree Reduction}

This subsection is devoted to showing two reduction rules for 2-layer fan-planar graphs.
These two reduction rules aim to reduce the maximum degree of a graph to a constant while preserving its 2-layer fan-planarity.

\begin{lemma}\label{lem:deg-1}
    Suppose that a vertex $v$ has degree-1 neighbors $u, w \in N(v)$.
    Then, $G$ has a 2-layer fan-planar drawing if and only if $G - w$ does.
\end{lemma}
\begin{proof}
    The only-if direction is clear since a subdrawing of a 2-layer fan-planar drawing is also 2-layer fan-planar.
    Hence, it suffices to show $G$ has a 2-layer fan-planar drawing, assuming that $G - w$ does.
    Let $\mathcal D$ be a 2-layer fan-planar drawing of $G - w$.
    Let $e_u$ and $e_w$ be the unique edges incident to $u$ and $w$, respectively.
    We extend the drawing $\mathcal D$ by inserting $w$ right next to $u$.
    We let $\mathcal D'$ be the obtained drawing of $G$.
    Since there are no other vertices between $u$ and $w$ in $\mathcal D'$, every edge that crosses $e_w$ in $\mathcal D'$ also crosses $e_u$ in $\mathcal D$.
    Thus, $\{e_w, f, f'\}$ is a violating triple in $\mathcal D'$ only if $\{e_u, f, f'\}$ is a violating triple in $\mathcal D$.
    Hence, $\mathcal D'$ is fan-planar.
\end{proof}

\begin{lemma}\label{lem:deg-3}
    Suppose that a vertex $v$ has at least five neighbors of degree at least~3.
    Then, $G$ has no 2-layer fan-planar drawing.
\end{lemma}
\begin{proof}
    Let $\mathcal D$ be an arbitrary 2-layer fan-planar drawing of $G$.
    Let $u_1, \dots, u_5$ be neighbors of $v$ with degree at least~3 that appear in this order in $\mathcal D$.
    We assume without loss of generality that $u_3$ has a neighbor $w \neq v$ with $w <_X v$.
    Since the degree of $u_2$ is at least~3, it has a neighbor $w' \notin \{v, w\}$.

    \begin{figure}[htb]
        \centering
        \begin{subfigure}[t]{0.3\textwidth}
            \centering
            \includegraphics[page=7]{figure.pdf}
            \subcaption{}
            \label{fig:deg-3:a}
        \end{subfigure}\hfill
        \begin{subfigure}[t]{0.3\textwidth}
            \centering
            \includegraphics[page=8]{figure.pdf}
            \subcaption{}
            \label{fig:deg-3:b}
        \end{subfigure}\hfill
        \begin{subfigure}[t]{0.3\textwidth}
            \centering
            \includegraphics[page=9]{figure.pdf}
            \subcaption{}
            \label{fig:deg-3:c}
        \end{subfigure}\hfill
        \caption{Subdrawings of $\mathcal D$ containing violating triples (bold edges).}
        \label{fig:deg-3}
    \end{figure}

    If $w'$ appears to the right of $v$, the edge $\{u_2, w'\}$ crosses both $\{u_3, w\}$ and $\{u_4, v\}$ (\zcref{fig:deg-3:a}).
    If $w'$ appears to the left of $w$, the edge $\{u_1, v\}$ crosses both $\{u_2, w'\}$ and $\{u_3, w\}$ (\zcref{fig:deg-3:b}).
    If $w'$ appears between $w$ and $v$, the edge $\{u_2, w'\}$ crosses both $\{u_1, v\}$ and $\{u_3, w\}$ (\zcref{fig:deg-3:c}).
    In all cases, $\mathcal D$ is not fan-planar.
\end{proof}

\begin{lemma}\label{lem:deg-2-matching}
    Suppose that a vertex $v$ has at least five degree-2 neighbors $u_1, \dots, u_5$ such that $G - v$ has a matching saturating $\{u_1, \dots, u_5\}$.
    Then, $G$ has no 2-layer fan-planar drawing.
\end{lemma}
\begin{proof}
    Let $\mathcal D$ be an arbitrary 2-layer fan-planar drawing of $G$.
    Assume that the vertices $u_1, \dots, u_5$ appear in this order in $\mathcal D$.
    For each $1 \le i \le 5$, let $w_i$ be the vertex that is matched to $u_i$.
    We assume without loss of generality that $w_3$ is to the left of $v$ in $\mathcal D$.

    \begin{figure}[htb]
        \centering
        \begin{subfigure}[t]{0.3\textwidth}
            \centering
            \includegraphics[page=1]{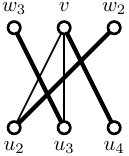}
            \subcaption{}
            \label{fig:matching:a}
        \end{subfigure}\hfill
        \begin{subfigure}[t]{0.3\textwidth}
            \centering
            \includegraphics[page=2]{figure.pdf}
            \subcaption{}
            \label{fig:matching:b}
        \end{subfigure}\hfill
        \begin{subfigure}[t]{0.3\textwidth}
            \centering
            \includegraphics[page=3]{figure.pdf}
            \subcaption{}
            \label{fig:matching:c}
        \end{subfigure}\hfill
        \caption{The cases in the proof of \zcref{lem:deg-2-matching}.}
        \label{fig:matching}
    \end{figure}

    If $w_2$ appears to the right of $v$, the edge $\{u_2, w_2\}$ crosses both $\{u_4, v\}$ and $\{u_3, w_3\}$ (\zcref{fig:matching:a}).
    If $w_2$ appears to the left of $w_3$, the edge $\{u_1, v\}$ crosses both $\{u_2, w_2\}$ and $\{u_3, w_3\}$ (\zcref{fig:matching:b}).
    If $w_2$ appears between $w_3$ and $v$, the edge $\{u_2, w_2\}$ crosses both $\{u_1, v\}$ and $\{u_3, w_3\}$ (\zcref{fig:matching:c}).
    In all cases, $\mathcal D$ is not fan-planar.
\end{proof}

\begin{lemma}\label{lem:deg-2}
    Suppose that a vertex $v$ has three degree-2 neighbors $u_1, u_2, u_3$ with $N(u_i) = \{v, w\}$ for $1 \le i \le 3$.
    Then, $G$ has a 2-layer fan-planar drawing if and only if $G - u_2$ does.
\end{lemma}
\begin{proof}
    It suffices to show that $G$ has a 2-layer fan-planar drawing, assuming that $G - u_2$ does.
    Let $\mathcal D$ be a 2-layer fan-planar drawing of $G - u_2$.
    Without loss of generality, we assume that $v, w \in X$, $v <_X w$, and $u_1 <_Y u_3$ in $\mathcal D$.

    We show that every edge crossing $\{v, u_3\}$ is incident to $w$.
    Suppose for a contradiction that there is an edge $e$ with $w \not\in e$ that crosses $\{v, u_3\}$.
    Since $N(u_1) = \{v, w\}$, $e$ is also not incident to $u_1$.
    Hence, edge $\{v, u_3\}$ crosses independent edges $e$ and $\{w, u_1\}$, which contradicts the fan-planarity of $\mathcal D$.
    By symmetry, every edge crossing $\{w, u_1\}$ is incident to $v$.

    Now we are ready to show that inserting $u_2$ anywhere between $u_1$ and $u_3$ in $\mathcal D$ yields a 2-layer fan-planar drawing of $G$.
    Suppose that a drawing $\mathcal D'$ obtained in this manner contains a violating triple $\{e, f, f'\}$ such that $e$ crosses independent edges $f, f'$.
    Since $\mathcal D$ is fan-planar, it contains either one of the new edges $\{v, u_2\}, \{w, u_2\}$.
    Observe that every edge crossing $\{v, u_2\}$ or $\{w, u_2\}$ also crosses $\{v, u_3\}$ or $\{w, u_1\}$.
    Consequently, if $e$ is the new edge, each of edges $e, f, f'$ is incident to $v$ or $w$, which implies that they are not independent.
    Hence, $f$ is the new edge.
    Without loss of generality, let us assume $f = \{v, u_2\}$.
    Then, $e$ is incident to $w$.
    Since $N(u_i) = \{v, w\}$ for $1 \leq i \leq 3$ and $e, f, f'$ are independent, $f' \cap \{v,w,u_1,u_3\} = \emptyset$.
    Now we have found another violating triple $\{e, \{v, u_3\}, f'\}$ without $u_2$, which contradicts the fan-planarity of $\mathcal D$.
    Hence, $\mathcal D'$ is indeed fan-planar.
\end{proof}

Now, we apply reduction rules based on \zcref{lem:deg-1,lem:deg-2}.
Suppose that a vertex $v$ has at least two degree-1 neighbors.
We remove all but one of them from $G$.
The safeness of the first reduction rule is shown in \zcref{lem:deg-1}.
Suppose that $v$ has at least three degree-2 neighbors $u_1, \dots, u_t$ with $N(u_i) = \{v, w\}$ for some $w$.
Then, we remove $u_i$ for all $3 \le i \le t$ from $G$.
The safeness of the second reduction rule is shown in \zcref{lem:deg-2}.
We say that $G$ is \emph{reduced} if none of these reduction rules can be applied to $G$.

\begin{lemma}\label{lem:degree-bound}
    Suppose that $G$ is reduced.
    If $G$ admits a 2-layer fan-planar drawing, then the maximum degree of $G$ is at most~13.
\end{lemma}
\begin{proof}
    Let $v$ be an arbitrary vertex in $G$.
    Without loss of generality, suppose that $v$ belongs to $X$.
    Since $G$ is reduced, $v$ has at most one degree-1 neighbor due to the first reduction rule based on \zcref{lem:deg-1}.
    In addition, $v$ has at most four neighbors of degree at least~3 due to \zcref{lem:deg-3}.
    To evaluate the number of degree-2 neighbors, let us define $Y_2(w)$ for each $w \in X \setminus \{v\}$ as the set $\{u \in Y \mid N(u) = \{v, w\}\}$.
    Due to \zcref{lem:deg-2-matching}, there are at most four vertices $w$ in $X \setminus \{v\}$ with $Y_2(w) \neq \emptyset$.
    Moreover, for every $w$, $|Y_2(w)| \le 2$ follows from the second reduction rule based on \zcref{lem:deg-2}.
    Thus, the degree of $v$ is at most $1 + 4  + 2 \cdot 4 = 13$.
\end{proof}

\subsection{Algorithm}

Let us now turn to our polynomial-time algorithm.
Our algorithm builds upon the dynamic programming algorithm for 2-layer $k$-planarity given in \cite{abs-2412-04042}\footnote{We use their corrected algorithm on arXiv~\cite[v4]{abs-2412-04042} since the original algorithm~\cite{KobayashiO025} contained an error.}, and we incorporate a condition for testing fan-planarity into this dynamic programming.

Let $\mathcal D$ be a 2-layer drawing of $G$.
For an integer $k \ge 0$, we say that $\mathcal D$ is \emph{$k$-planar} if each edge involves at most $k$ crossings.
The following observation is crucial for our algorithm.

\begin{observation}\label{obs:k-planar}
    Let $G$ be a bipartite graph with degree at most $d$.
    Then, every 2-layer fan-planar drawing of $G$ is $d$-planar.
\end{observation}
\begin{proof}
    Let $\mathcal D$ be a 2-layer fan-planar drawing of $G$.
    For each edge $e$, the crossing edges in $\cross_{\mathcal D}(e)$ have a common end vertex $v$.
    Since the degree of $v$ is at most $d$, there are at most $d$ edges in $\cross_{\mathcal D}(e)$.
\end{proof}

Given this, by~\zcref{lem:degree-bound}, every 2-layer fan-planar drawing of a reduced bipartite graph $G$ is $13$-planar, which allows us to search for a 2-layer fan-planar drawing among 2-layer $13$-planar drawings.
In the remaining part, we first give a sketch of a polynomial-time algorithm of~\cite{abs-2412-04042} for deciding whether a bipartite graph admits a 2-layer $k$-planar drawing for fixed~$k$, and then explain how to modify the algorithm for our purpose.

\paragraph*{Sketch of the Algorithm.}

Let $G = (X \cup Y, E)$ be a bipartite graph and let $k \geq 0$ be a fixed integer.
The algorithm first applies degree-reduction rules to ensure that $G$ has degree at most $2k + 2$.
However, we do not apply them since we set $k$ to be not less than the degree of a reduced bipartite graph.
We can assume that $G$ is connected.
We can also assume $|X| > 4k + 1$, as otherwise $G$ has a constant number of vertices combined with the above degree bound.
Under these assumptions, the algorithm performs a dynamic programming.

Let $\mathcal D^*$ be a 2-layer $k$-planar drawing of $G$ and let $S \subseteq X$ with $|S| = 4k + 2$ such that the vertices in $S$ appear consecutively in $\mathcal D^*$.
Let $\mathcal D$ be the subdrawing of $\mathcal D^*$ induced by $N[S]$.
A crucial observation \cite[Lemma 8]{KobayashiO025} is that there is no path between $u \in X \setminus S$ appearing to the left of $S$ and $v \in X \setminus S$ appearing to the right of $S$ in $G -  N[S]$, that is, $N[S]$ separates the drawing $\mathcal D^*$ into the ``left part'' and the ``right part'' such that there is no path between these parts in $G - N[S]$.
Moreover, each edge contained in the ``left part'' does not cross any edge contained in the ``right part''~\cite[Lemma 4]{KobayashiO025}.
These observations suggest the following dynamic programming algorithm.
For $S \subseteq X$ with $|S| = 4k + 2$, a 2-layer $k$-planar drawing $\mathcal D = (\sigma_X, \sigma_Y)$ of $G[N[S]]$, a function $\cross: \delta(S) \to \{0, 1, \dots, k\}$, and a set $\mathcal{C}$ of components of $G - N[S]$, $\draw(S, \mathcal D, \cross, \mathcal C)$ is a Boolean predicate that is true if and only if $G[L]$ admits a 2-layer $k$-planar drawing $\mathcal D_L = (\tau_X, \tau_Y)$ such that
\begin{itemize}
    \item $\mathcal D$ is the subdrawing of $\mathcal D_L$ induced by $N[S]$,
    \item $\sigma_X$ is a suffix of $\tau_X$, and
    \item for every edge $e \in \delta(S)$, $\cross(e) = |\cross_{\mathcal D_L}(e)|$,
\end{itemize}
where $L = N[S] \cup (\bigcup_{C \in \mathcal{C}} C)$.
Such a drawing $\mathcal D_L$ is called an \emph{$(S, \mathcal D, \cross, \mathcal C)$-drawing} of $G[L]$.
Clearly, $G$ admits a 2-layer $k$-planar drawing if and only if there are some $S, \mathcal D, \cross, \mathcal{C}$ such that $\draw(S, \mathcal D, \cross, \mathcal{C})$ is true and $L = V(G)$.
The algorithm~\cite{abs-2412-04042} then evaluates $\draw(S, \mathcal D, \cross, \mathcal{C})$ for each $S, \mathcal D, \cross, \mathcal{C}$ in the following recursive manner, where $v$ denotes the last (rightmost) vertex on $\sigma_X$; see \zcref{fig:illustration-of-the-dp} for an illustration.

\begin{figure}[ht]
    \centering
    \includegraphics[page=11]{figure.pdf}
    \caption{An illustration of the dynamic programming algorithm for 2-layer $k$-planar graphs.}
    \label{fig:illustration-of-the-dp}
\end{figure}

\begin{enumerate}
    \item\label{item:rec-step-first} If $\mathcal{C} = \emptyset$, since $|L| = O(k^2)$, compute $\draw(S, \mathcal D, \cross, \mathcal{C})$ na\"ively.

    \item  If $v$ has a neighbor $y$ that has another neighbor in $L$ but none of them is one of the last $2k+2$ vertices of $\sigma_X$, then $\draw(S, \mathcal D, \cross, \mathcal{C})$ is false.

    \item If there is an edge $e \in \delta(v)$ such that $\cross(e) \neq |\cross_{\mathcal D}(e)|$, then $\draw(S, \mathcal D, \cross, \mathcal{C})$ is false.

    \item\label{item:rec-step-last} Otherwise, we compute the value with the following recurrence:
          \[
              \draw(S, \mathcal D, \cross, \mathcal{C}) = \bigvee_{S', \mathcal D', \cross', \mathcal{C}'} \draw(S', \mathcal D', \cross', \mathcal{C}'),
          \]
          where $(S', \mathcal D', \cross', \mathcal{C}')$ is taken over all possible combinations such that
          \begin{itemize}
              \item $S'$ is obtained from $S$ by removing $v$ and adding a vertex $u$ in $(L \cap X) \setminus S$,
              \item $\mathcal D'$ is a 2-layer $k$-planar drawing of $G[N[S']]$ such that $u$ is the leftmost vertex on $X$ in $\mathcal D'$ and no pair of vertices on the same layer appears in opposite orders in $\mathcal D$ and $\mathcal D'$,
              \item $\cross'$ is a function $\delta(S') \to \{0, \dots, k\}$ satisfying $\cross'(e) = \cross(e) - |\cross_{\mathcal D}(e) \cap \delta(v)|$ for all $e \in \delta(S') \cap \delta(S)$, and
              \item $\mathcal C'$ is the set of components of $G - N[S']$ that intersect with a component in $\mathcal C$.
          \end{itemize}
\end{enumerate}

We refer the reader to~\cite{abs-2412-04042} for the correctness of the algorithm.
Additionally, the paper~\cite{abs-2412-04042} shows that the number of possible tuples $(S, \mathcal D, \cross, \mathcal{C})$ and the total running time to evaluate the recurrence are both bounded by $2^{O(k^3)}n^{4k + O(1)}$.

\paragraph*{Incorporating Fan-Planarity into the Algorithm.}

Now, we extend this recurrence.
Let us similarly define a Boolean predicate $\drawfan(S, \mathcal D, \cross, \mathcal{C})$ to be true if and only if the induced subgraph $G[L]$ admits an $(S, \mathcal D, \cross, \mathcal{C})$-drawing that is fan-planar.
We claim that this can be computed in the same manner just by additionally requiring each $\mathcal D$ to be fan-planar.
Namely, we set $\drawfan(S, \mathcal D, \cross, \mathcal{C})$ to false whenever $\mathcal D$ is not fan-planar, and otherwise, compute the value following the steps in the above exactly, replacing $\draw$ with $\drawfan$.

It is clear that this algorithm evaluates the value of some $\drawfan(S, \mathcal D, \cross, \mathcal{C})$ with $L = V(G)$ to be true if $G$ admits a 2-layer $k$-planar drawing that is fan-planar.
Conversely, following the correctness of the algorithm for testing 2-layer $k$-planarity, it is not difficult to see that if the value of some $\drawfan(S, \mathcal D, \cross, \mathcal{C})$ with $L = V(G)$ is evaluated to be true it ensures the existence of a 2-layer $k$-planar drawing $\mathcal{D}^* = (\sigma_X, \sigma_Y)$ of $G$ such that any $4k+2$ consecutive vertices on $\sigma_X$ and their neighbors induce a fan-planar drawing.\footnote{Informally, the original dynamic programming algorithm does not search for a specific type of 2-layer $k$-planar drawings: it ``accepts'' every 2-layer $k$-planar drawing. Hence, the modified algorithm simply searches for all 2-layer $k$-planar drawings with this property.}
In fact, this additional constraint indeed implies the fan-planarity of $\mathcal{D}^*$.

\begin{lemma}
    Let $\mathcal D^* = (\sigma_X, \sigma_Y)$ be a 2-layer $k$-planar drawing of a connected bipartite graph $G$ with $|X| \geq 4k + 2$.
    Then, $\mathcal D^*$ is fan-planar if and only if, for every set $S \subseteq X$ of $4k+2$ consecutive vertices on $\sigma_X$, the subdrawing of $\mathcal D^*$ induced by $N[S]$ is fan-planar.
\end{lemma}

\begin{proof}
    The only-if direction is clear.
    For the if direction, suppose for a contradiction that $\mathcal D^*$ contains a violating triple $\{e, f, f'\}$ even though $\mathcal D^*$ satisfies the latter condition.
    Let $x_e, x_f, x_{f'}$ be the endpoints of $e, f, f'$ in $X$, respectively.
    Let $S \subseteq X$ be a set of $4k+2$ consecutive vertices on $\sigma_X$ containing the vertices at most $2k$ apart from $x_e$, namely, the vertices $x'$ such that there are at most $2k-1$ vertices on $\sigma_X$ between $x_e$ and $x'$, excluding them.
    In a 2-layer $k$-planar drawing, two edges do not cross if their endpoints are $2k$ apart~\cite{abs-2412-04042}.
    This implies that $x_f, x_{f'} \in S$ and hence $e, f, f' \subseteq N[S]$, which contradicts that the subdrawing of $\mathcal D^*$ induced by $N[S]$ is fan-planar.
\end{proof}

Hence, our modification correctly computes $\drawfan(S, \mathcal D, \cross, \mathcal{C})$.
Since checking the fan-planarity of each $\mathcal D$ additionally only takes $k^{O(1)}$ time, it does not affect the bound on running time.
This leads to the following lemma.

\begin{lemma}\label{lem:alg}
    For fixed $k \geq 0$, there is a $2^{O(k^3)}n^{4k+O(1)}$-time algorithm that, given a bipartite graph $G$ with degree at most $k$, decides whether $G$ admits a 2-layer $k$-planar drawing that is fan-planar.
\end{lemma}

For a given graph $G$, it is clear that we can repeatedly apply the reduction rules and obtain a reduced graph in polynomial time.
Hence, by combining \zcref{lem:degree-bound,lem:alg} and \zcref{obs:k-planar}, we obtain an algorithm for \zcref{thm:main}.

\section{One-Sided Fan-Planarity}\label{sec:one-sided}

In this section, we prove \zcref{thm:one-sided} following the idea in the previous section.
We first reduce the maximum degree of the input graph analogously.
We then give a linear-time algorithm for solving the problem on instances with bounded degree.
Note that, after reducing the maximum degree (and eliminating isolated vertices), following the dynamic programming of Odak, Rollin, and Scheele~\cite{OdakRS026} for one-sided $k$-matching-planarity already yields a linear-time algorithm.
However, even for $k=1$, their dynamic programming tries all permutations of $N(X_i)$ for each set $X_i$ of $\ell + 1 = 52$ vertices consecutive on $\sigma_X$, which makes the table size at least $52! \cdot |X| > 2^{225.58} |X|$ asymptotically.
Moreover, the dynamic programming algorithm for one-sided $k$-planarity by Kobayashi, Okada, and Wolff~\cite{KobayashiO025} can be applied in a similar manner, but its running time involves a huge constant factor for the same reason.
Hence, we tailor another linear-time dynamic-programming algorithm for \zcref{thm:one-sided} to avoid maintaining permutations and reduce the table size to at most $\OneSidedTableConstantFactor \cdot |Y| < 2^{12.01} |Y|$.

\subsection{Degree Reduction}

We reduce the maximum degree of the input graph $G$ with a similar set of reduction rules.
Note that it is not straightforward to apply exactly the same set of reduction rules used in the previous section, since we have to comply with the prescribed order $\sigma_X$ for $X$.
However, we can still safely apply \zcref{lem:deg-3,lem:deg-2-matching}: If $G$ has a vertex $v$ having at least five neighbors of degree at least~3 (\zcref{lem:deg-3}) or having at least five degree-2 neighbors saturated by a matching in $G - v$ (\zcref{lem:deg-2-matching}), we can immediately conclude that $G$ has no 2-layer fan-planar drawing with respect to $\sigma_X$.
Hence, it suffices to give the one-sided counterparts of \zcref{lem:deg-1,lem:deg-2}.
Let us begin with the one-sided counterpart of \zcref{lem:deg-1}.
For convenience, let us use $\sigma_X - v$ as shorthand for the linear order obtained from $\sigma_X$ by removing $v$ for $v \in X$.

\begin{lemma}\label{lem:one-sided:deg-1xy}
    (1) Suppose that a vertex $v \in X$ has degree-1 neighbors $u_1, u_2, u_3 \in N(v)$.
    Then, $G$ has a 2-layer fan-planar drawing with $\sigma_X$ if and only if $G - u_2$ has one respecting $\sigma_X$.
    (2) Suppose that a vertex $v \in Y$ has degree-1 neighbors $u_1, u_2, u_3 \in N(v)$ with $u_1 <_X u_2 <_X u_3$.
    Then, $G$ has a 2-layer fan-planar drawing with $\sigma_X$ if and only if $G - u_2$ has one respecting $\sigma_X - u_2$.
\end{lemma}

\begin{proof}
    The only-if directions are clear since a subdrawing of a 2-layer fan-planar drawing is also 2-layer fan-planar.
    Suppose that $G - u_2$ has a 2-layer fan-planar drawing $\mathcal D$ with $\sigma_X$ (or $\sigma_X - u_2$ when $u_2 \in X$).
    It suffices for both cases to show that inserting $u_2$ anywhere between $u_1$ and $u_3$ in $\mathcal D$ does not break its fan-planarity.
    Without loss of generality, let us assume $u_1 <_Y u_3$ if $v \in X$.
    Let us suppose, for the sake of contradiction, that there is a 2-layer drawing $\mathcal D'$ of $G$ that is obtained in such a way but is not fan-planar.
    Since $e_2 = \{v, u_2\}$ is the only edge newly added, this edge must be involved in a violating triple.

    Let $f, f' \in \cross_{\mathcal D'}(e_2)$ form a violating triple $\{e_2, f, f'\}$ in $\mathcal D'$.
    Let $w$ and $w'$ be the end vertices of $f$ and $f'$ that belong to the same bipartition class as $u_2$, respectively.
    Suppose that both $w$ and $w'$ are either to the left of $u_2$ or to the right of $u_2$.
    Without loss of generality, assume that both $w$ and $w'$ are to the left of $u_2$.
    Then, both $f$ and $f'$ cross $e_3 = \{v, u_3\}$, which forms a violating triple $\{e_3, f, f'\}$, a contradiction.

    Suppose the other case; assume without loss of generality that $w$ is to the left of $u_2$ and $w'$ is to the right of $u_2$ in $\mathcal D'$.
    Since $f$ crosses $e_2$ in $\mathcal D'$, it also crosses $e_3 = \{v, u_3\}$.
    Similarly, $f'$ crosses $e_1 = \{v, u_1\}$.
    This implies that $\{e_1, f, f'\}$ (and $\{e_3, f, f'\}$) is a violating triple in $\mathcal D$, contradicting the assumption.
\end{proof}

Next we give the one-sided counterpart of \zcref{lem:deg-2}.
The proof of the lemma is, in fact, identical to \zcref{lem:deg-2}; the proof shows that it is ``safe'' to insert $u_2$ anywhere between $u_1$ and $u_3$ in a 2-layer fan-planar drawing of $G - u_2$.

\begin{lemma}\label{lem:one-sided:deg-2xy}
    (1) Suppose that a vertex $v \in X$ has three degree-2 neighbors $u_1, u_2, u_3 \in Y$ with $N(u_i) = \{v, w\}$ for $1 \le i \le 3$.
    Then, $G$ has a 2-layer fan-planar drawing with $\sigma_X$ if and only if $G - u_2$ has one with $\sigma_X$.
    (2) Suppose that a vertex $v \in Y$ has three degree-2 neighbors $u_1, u_2, u_3 \in X$ with $u_1 <_X u_2 <_X u_3$ and $N(u_i) = \{v, w\}$ for $1 \le i \le 3$.
    Then, $G$ has a 2-layer fan-planar drawing with $\sigma_X$ if and only if $G - u_2$ has one with $\sigma_X - u_2$.
\end{lemma}

Similarly, we then apply reduction rules based on \zcref[range]{lem:one-sided:deg-1xy,lem:one-sided:deg-2xy} together with \zcref{lem:deg-3,lem:deg-2-matching}.
Almost the same estimation as \zcref{lem:degree-bound} gives us an upper bound $14$ on the maximum degree after exhaustively applying these reduction rules.
However, it is not trivial whether the reduction can be done in linear time to eventually obtain a reduced graph.
In the following, we carefully analyze to show that it suffices to apply the reduction rules for only one ``round'', not exhaustively, to achieve the same upper bound.

\begin{lemma}\label{lem:one-sided:degree-bound-linear-time}
    Given a bipartite graph $G = (X \cup Y, E)$ and a linear order $\sigma_X$ on $X$, in linear time one can either conclude that there is no 2-layer fan-planar drawing of $G$ respecting $\sigma_X$, or compute a bipartite graph $G'$ of maximum degree at most 14 and linear order $\sigma_{X'}$ on $X'$ such that $G$ admits a 2-layer fan-planar drawing respecting $\sigma_X$ if and only if $G'$ admits one respecting $\sigma_{X'}$.
\end{lemma}

\begin{proof}
    We claim the following algorithm satisfies the requirements.
    These steps can be simulated in (worst-case) linear time using radix sort, maintaining $X'$ and $Y'$ by Boolean arrays.

    \begin{enumerate}
        \item Let $V_i$ be the set of vertices of degree $i$ for $0 \leq i \leq 2$ and let $V_{\geq 3} = (X \cup Y) \setminus (V_0 \cup V_1 \cup V_2)$.
        \item\label{step:deg-3} (\zcref{lem:deg-3}) If there is a vertex $v$ with $|N(v) \cap V_{\geq 3}| \geq 5$, conclude that there is no 2-layer fan-planar drawing of $G$ respecting $\sigma_X$.
        \item\label{step:deg-2-matching} (\zcref{lem:deg-2-matching}) If there is a vertex $v$ with $|N(N(v) \cap V_2) \setminus \{v\}| \geq 5$, conclude that there is no 2-layer fan-planar drawing of $G$ respecting $\sigma_X$.
        \item Set $X' \leftarrow X$ and $Y' \leftarrow Y$.
        \item\label{step:reduce-twins} (\zcref[range]{lem:one-sided:deg-1xy,lem:one-sided:deg-2xy}) Group the vertices in $V_1 \cup V_2$ into vertex sets according to their neighborhoods, that is, a group is defined as a maximal set $S \subseteq V_1 \cup V_2$ such that $N(u) = N(v)$ for $u, v \in S$. For each group $S$ having more than two vertices, do the following:
              \begin{itemize}
                  \item if $S \subseteq X$, let $s_{\min}$ and $s_{\max}$ be the minimum and maximum vertices in $S$ with respect to $\sigma_X$, and set $X' \leftarrow X' \setminus (S \setminus \{s_{\min}, s_{\max}\})$.
                  \item otherwise, namely if $S \subseteq Y$, let us take arbitrarily $S' \subseteq S$ with $|S'| = |S| - 2$ and set $Y' \leftarrow Y' \setminus S'$.
              \end{itemize}
        \item Let $G'$ be the subgraph of $G$ induced by $X' \cup Y'$ and $\sigma_{X'}$ be the linear order obtained from $\sigma_X$ by restricting the underlying set to $X'$.
    \end{enumerate}

    The correctness of each step follows from the corresponding lemma(s), and it suffices to show that the maximum degree of $G'$ is at most 14.
    Although the proof mostly follows that of \zcref{lem:degree-bound}, we classify and count the neighbors by their original degrees; let $v$ be an arbitrary vertex in $G'$ and we separately bound each of $|N_{G'}(v) \cap V_0|$, $|N_{G'}(v) \cap V_1|$, $|N_{G'}(v) \cap V_2|$, and $|N_{G'}(v) \cap V_{\geq 3}|$.
    \begin{itemize}
        \item Clearly, $|N_{G'}(v) \cap V_0| = 0$.
        \item Due to step~\ref{step:reduce-twins}, at most two neighbors of $v$ in $V_1$ are in $G'$. Hence, $|N_{G'}(v) \cap V_1| \leq 2$.
        \item Due to step~\ref{step:deg-2-matching}, it holds that $|N_{G'}(N_{G'}(v) \cap V_2) \setminus \{v\}| \leq 4$ since $G'$ is a subgraph of $G$. For each $w \in (N_{G'}(N_{G'}(v) \cap V_2) \setminus \{v\})$, there are at most two vertices from $V_2$ remaining in $G'$ that are adjacent to both $v$ and $w$ due to step~\ref{step:reduce-twins}. Hence, $|N_{G'}(v) \cap V_2| \leq 4 \cdot 2 = 8$.
        \item Since $G'$ is a subgraph of $G$, $|N_{G'}(v) \cap V_{\geq 3}| \leq |N(v) \cap V_{\geq 3}| \leq 4$ by step~\ref{step:deg-3}.
    \end{itemize}
    Overall, $|N_{G'}(v)| \leq |N_{G'}(v) \cap V_0| + |N_{G'}(v) \cap V_1| + |N_{G'}(v) \cap V_2| + |N_{G'}(v) \cap V_{\geq 3}| \leq 14$.
\end{proof}

\subsection{Algorithm}

By \zcref{lem:one-sided:degree-bound-linear-time}, let us assume that $G$ has maximum degree at most 14.
Let us also assume that $G$ has no isolated vertex, as it can be safely removed.
Note that we cannot trivially assume that $G$ is connected.
Moreover, let us assume that every vertex in $G$ has at most two degree-1 neighbors by \zcref{lem:one-sided:deg-1xy}.
We also reapply the reduction rule based on \zcref{lem:one-sided:deg-2xy}~(1) once to make sure that for each pair of $x, x' \in X$ there are at most two vertices $y, y'$ such that $N(y) = N(y') = \{x, x'\}$.
This can be done in linear time as shown in \zcref{lem:one-sided:degree-bound-linear-time}.
Moreover, this reduction yields neither a new degree-2 vertex in $Y$ nor a new degree-1 vertex.

Similarly to the algorithms for one-sided $k$-planarity~\cite{KobayashiO025} and one-sided $k$-matching-planarity~\cite{OdakRS026}, we design a dynamic-programming algorithm that constructs a desired drawing from the left.
However, unlike those algorithms, instead of sliding the window on the fixed side $\sigma_X$, our algorithm guesses a next vertex on $\sigma_Y$ one by one.

Let $\mathcal D$ be a 2-layer drawing of $G$ respecting $\sigma_X$.
For $y \in Y$, let $x_\ell(y)$ and $x_r(y)$ be the leftmost and rightmost vertices in $N(y)$, respectively, and let $e_\ell(y) = \{x_\ell(y), y\}$ and $e_r(y) = \{x_r(y), y\}$.
Note that $x_\ell(y) = x_r(y)$ when $|N(y)| = 1$.

\begin{lemma}\label{lem:x_ell-x_r-distance}
    If $\mathcal D$ is fan-planar, there are at most $40$ vertices between $x_\ell(y)$ and $x_r(y)$ for every $y \in Y$.
\end{lemma}

\begin{proof}
    Between $x_\ell(y)$ and $x_r(y)$, there are at most $12$ neighbors of $y$.
    Except for those neighbors of $y$, every edge incident to a vertex between them crosses either $e_\ell(y)$ or $e_r(y)$.
    Since $\mathcal D$ is $14$-planar due to \zcref{obs:k-planar} and there are no isolated vertices, there are at most $12 + 14 \cdot 2 = 40$ vertices between them.
\end{proof}

By \zcref{lem:x_ell-x_r-distance}, we can immediately conclude that the instance~$(G, \sigma_X)$ is infeasible when it contains more than $40$ vertices between $x_\ell(y)$ and $x_r(y)$ for some $y \in Y$.
Next, we give a necessary and sufficient local condition, which is the key point of our algorithm.

\begin{observation}\label{obs:condition-around-y-in-Y}
    If $\mathcal D$ is fan-planar, the subdrawing of $\mathcal D$ induced by $\delta(y) \cup \cross_{\mathcal D}(e_\ell(y)) \cup \cross_{\mathcal D}(e_r(y))$ is fan-planar for every $y \in Y$.
\end{observation}

The above observation is clear, as every subgraph of a fan-planar drawing is also fan-planar.
The following lemma proves that the converse also holds.

\begin{lemma}\label{lem:conditions-around-y-in-Y-suffices}
    If the condition in \zcref{obs:condition-around-y-in-Y} holds for every $y \in Y$, $\mathcal D$ is fan-planar.
\end{lemma}

\begin{proof}
    Suppose that $\mathcal D$ is not fan-planar.
    Then, there exists an edge $e \in E$ crossing two independent edges $f, f'$ in $\mathcal D$.
    Let $y \in Y$ be the end vertex of $e$ in $Y$.
    Now observe that every edge crossing $e$ also crosses either $e_\ell(y)$ or $e_r(y)$.
    Hence, $\delta(y) \cup \cross_{\mathcal D}(e_\ell(y)) \cup \cross_{\mathcal D}(e_r(y))$ contains all of $e, f$ and $f'$, which implies that the subdrawing  induced by $\delta(y) \cup \cross_{\mathcal D}(e_\ell(y)) \cup \cross_{\mathcal D}(e_r(y))$ is not fan-planar.
\end{proof}

This local condition suggests a dynamic programming algorithm that maintains the current last vertex $y$ and the set $L \subseteq Y \setminus \{y\}$ of vertices to the left of $y$.
Observe that $L$ determines $\cross_{\mathcal D}(e_\ell(y))$ and $\cross_{\mathcal D}(e_r(y))$ regardless of how we order the vertices inside $L$ and $Y \setminus (L \cup \{y\})$ in $\mathcal D$.
More precisely, we have
\begin{align*}
    \cross_{\mathcal D}(e_\ell(y)) & = \{\{x', y'\} \in E \mid (y' \in L \land x_{\ell}(y) <_X x') \lor (y' \in Y \setminus (L \cup \{y\}) \land x' <_X x_\ell(y))\}, \\
    \cross_{\mathcal D}(e_r(y))    & = \{\{x', y'\} \in E \mid (y' \in L \land x_r(y) <_X x') \lor (y' \in Y \setminus (L \cup \{y\}) \land x' <_X x_r(y))\}.
\end{align*}
We define the Boolean value $\mathrm{dp}(y, L)$ to be true if there is a 2-layer (possibly non-fan-planar) drawing $\mathcal D$ of the whole graph $G$ respecting $\sigma_X$ such that the set of vertices to the left of $y$ is exactly $L$ and each vertex in $L \cup \{y\}$ satisfies the condition in \zcref{obs:condition-around-y-in-Y}.
Let $Q(y, L)$ be the Boolean value that is true if and only if $\cross_{\mathcal D}(e_\ell(y))$ and $\cross_{\mathcal D}(e_r(y))$ determined by $L$ satisfy the condition in \zcref{obs:condition-around-y-in-Y}.
For the base case $L = \emptyset$, $\mathrm{dp}(y, \emptyset) = Q(y, \emptyset)$.
Otherwise, we can compute the value of $\mathrm{dp}(y, L)$ with the following recurrence
\[
    \mathrm{dp}(y, L) = Q(y, L) \land \bigvee_{y', L'} \mathrm{dp}(y', L'),
\]
where $(y', L')$ is taken over the possible combinations such that $y' \in L$ and $L' = L \setminus \{y'\}$.
By \zcref{lem:conditions-around-y-in-Y-suffices}, if there is a 2-layer fan-planar drawing of $G$ respecting $\sigma_X$, there exists $y \in Y$ such that $\mathrm{dp}(y, Y \setminus \{y\})$ is true.
Conversely, the existence of such $y$ implies the existence of a 2-layer fan-planar drawing of $G$ respecting $\sigma_X$.

Although a na\"ive evaluation of this recurrence would require $\Omega(2^{|Y|})$ time, this is the basis of our linear-time algorithm.
First, to achieve a linear time bound, observe that the set $L$ of vertices to the left of $y$ in $\mathcal D$ is determined by the combination of $y$, the set  $\cross_{\mathcal D}(e_\ell(y))$ of crossing edges with $e_\ell(y)$, and (at most two) degree-1 neighbors of $x_\ell(y)$.
Due to \zcref{obs:k-planar}, this directly yields a polynomial upper bound on the number of possible sets $L$ that appear in evaluating the above recurrences.

\begin{observation}\label{obs:restoring-L}
    Let $y \in Y$, $F \subseteq E$, and $A$ be a subset of degree-1 neighbors of $x_{\ell}(y)$.
    Let $\mathcal D$ be a 2-layer drawing of $G$ such that
    \begin{itemize}
        \item $\cross_{\mathcal D}(e_{\ell}(y)) = F$ and
        \item the set of degree-1 neighbors of $x_{\ell}(y)$ that are to the left of $y$ is $A$.
    \end{itemize}
    Then, the set $L \subseteq Y$ of vertices that appear to the left of $y$ in $\mathcal D$ is determined by
    \begin{align*}
        L =  \{y' \in Y \mid \{x', y'\} \in F, \, x_\ell(y) <_X x'\} \cup \{y' \in Y \mid \{x', y'\} \in E \setminus F, \, x' < x_{\ell}(y) \} \cup A.
    \end{align*}
\end{observation}

\begin{proof}
    Let $y' \in Y \setminus \{y\}$.
    Since $G$ has no isolated vertices, there is an edge $e$ incident to $y'$.
    Let $x'$ be the other end vertex of $e$.
    If $N(y') = \{x'\} = \{x_\ell(y)\}$, then $y' \in L$ if and only if $y' \in A$.
    In the following, we assume that $x' \neq x_\ell(y)$.
    Suppose that $e \in F$.
    Then, $y'$ appears to the left of $y$ if and only if $x_{\ell}(y) <_X x'$.
    Suppose that $e \in E \setminus F$.
    Then, $y'$ appears to the left of $y$ if and only if $x' < x_{\ell}(y)$, as otherwise $e$ crosses $e_{\ell}(y)$.
    Hence, the set of vertices that appear to the left of $y$ in $\mathcal D$ is equal to $L$.
\end{proof}

Let $p(y, L)$ be the tuple consisting of $y$, $F \subseteq E$ and $A \subseteq N_\ell(y)$, which determines $L$ as in \zcref{obs:restoring-L}.
As $|F| \le 14$ and $|A| \le 2$, the number of possible sets $L$ represented by these tuples is upper-bounded by a polynomial in the input size.
Next, we give a lemma that further reduces the table size to a linear size with an explicit constant factor.

\begin{lemma}\label{lem:few-candidates-for-each-y}
    For each $y \in Y$, let $P(y)$ be the set consisting of all $p(y, L)$ such that there is a 2-layer fan-planar drawing $\mathcal D$ of $G$ respecting $\sigma_X$ and $L$ being the set of vertices to the left of $y$.
    Given $y$, one can construct in constant time a set $P'(y)$ such that $P(y) \subseteq P'(y)$ and $|P'(y)| \leq \OneSidedTableConstantFactor$.
\end{lemma}

\begin{proof}
    We first show the upper bound $|P(y)| \le \OneSidedTableConstantFactor$.
    Let $\mathcal D$ be a 2-layer fan-planar drawing of $G$ such that $\cross_{\mathcal D}(e_\ell(y)) = F$ and the set of degree-1 neighbors of $x_{\ell}(y)$ that are to the left of $y$ is $A$ for some $(y, F, A) \in P(y)$.
    Note that all edges in $F$ are incident to a single vertex $v$ due to the fan-planarity of $\mathcal D$.
    Therefore, we say that $v$ \emph{covers} $F$, namely $F \subseteq \delta(v)$.

    First suppose that $F \neq \emptyset$.
    Let $e' = \{x', y'\}$ be an arbitrary edge in $F$ with $x' \in X$ and $y' \in Y$.
    Then, $x'$ or $y'$ covers $F$.
    Due to the 14-planarity of $\mathcal D$, each of $e'$ and $e_{\ell}(y)$ involves at most 14 crossings and hence crosses at most 13 other edges.
    Since every edge incident to a vertex between $x'$ and $x_\ell(y)$ must cross $e'$ or $e_\ell(y)$, the sum of the degrees of vertices between them is at most $13 \cdot 2 = 26$.
    Thus, as $G$ has no isolated vertices, there are at most 26 vertices between $x'$ and $x_{\ell}(y)$.
    This implies that candidates for $x'$ are placed within 27 vertices to the left or right of $x_{\ell}(y)$.
    Hence, there are at most $27 \cdot 2 = 54$ candidates for $x'$.
    Moreover, this also implies that there are at most $(14 + 26) \cdot 2 = 80$ candidates for $e' = \{x', y'\}$: edge $e'$ is either incident to the leftmost or rightmost vertices satisfying the above degree-sum condition, or incident to vertices in-between.
    Hence, there are at most $80$ candidates for $y'$.

    Suppose that $y'$ covers $F$.
    Then $F$ is determined by whether $y'$ lies to the left or right of $y$.
    Hence, in this case, there are only two possible $F$ for each $y'$.

    Suppose that $x'$ covers $F$ and let $y'' \in N(x')$.
    If $y''$ has a neighbor $x''$ that is neither $x'$ nor $x_\ell(y)$, the positions of $x''$ and $x_\ell(y)$ determine if $y'' \in L$ since $\{y'', x''\} \not\in F$.
    Note that this then determines if $\{x', y''\} \in F$.
    There are at most four neighbors that do not have such a neighbor: at most two degree-1 vertices and at most two degree-2 vertices that are adjacent to $x'$ and $x_\ell(y)$.
    Hence, in this case at most four neighbors of $x'$ effectively determine $F$ and there are $16$ possibilities for each $x'$.

    The remaining case is $F = \emptyset$ and there is clearly one candidate for $F$.

    Overall, there are at most $54 \cdot 16 + 80 \cdot 2 + 1 = 1025$ candidates for $F$.
    Combined with the four possibilities of $A$, there are at most $1025 \cdot 4 = \OneSidedTableConstantFactor$ candidates for $(y, F, A)$.
    Moreover, observe that all of those candidates can be enumerated locally in constant time accordingly.
\end{proof}

Given \zcref{lem:few-candidates-for-each-y}, the following is an outline of our linear-time algorithm.
\begin{enumerate}
    \item Compute $\mathcal P = \bigcup_{y \in Y} P'(y)$ in linear time by \zcref{lem:few-candidates-for-each-y}.
    \item\label{step:removing-invalid-candidates} Remove the \emph{invalid} candidates in $\mathcal P$; namely, remove every $(y, F, A) \in \mathcal P$ such that either (1) there is no corresponding pair $(y, L)$ with $(y, F, A) = p(y, L)$ or (2) there is a corresponding pair $(y, L)$ but the set of edges crossing $e_\ell(y)$ determined by $(y, L)$ differs from $F$.
    \item\label{step:removing-trivially-no-candidates} Remove every $p(y, L) = (y, F, A) \in \mathcal P$ such that $Q(y, L)$ is false.
    \item\label{step:compute-dp-predicates} For $p(y, L) = (y, F, A) \in \mathcal P$, define $\mathrm{dp}(y, F, A) = \mathrm{dp}(y, L)$. For each $(y, F, A)$, compute a set $\mathcal{P}' \subseteq \mathcal{P}$ with $|\mathcal{P}'| \leq 69$. We will later show that this set $\mathcal P'$ is sufficient for evaluating the recurrence
          \[
              \mathrm{dp}(y, F, A) = \bigvee_{\substack{(y', F', A') \in \mathcal{P}'}} \mathrm{dp}(y', F', A').
          \]
    \item\label{step:dp-over-P} Compute the value of $\mathrm{dp}(y, F, A)$ for every $(y, F, A) \in \mathcal P$ by the above recurrence.
\end{enumerate}
Assuming we can compute such $\mathcal P'$ in step \zref{step:compute-dp-predicates}, the correctness of this algorithm is almost clear.
We explain how to implement the latter four steps in linear time.

\paragraph{Steps \zref{step:removing-invalid-candidates} and \zref{step:removing-trivially-no-candidates} in Linear Time.}
At first, for each $(y, F, A) \in \mathcal P$, we verify that we can actually determine the set $L$ of vertices to the left of $y$ from $(y, F, A)$; by \zcref{obs:restoring-L}, otherwise, $(y, F, A)$ is an invalid candidate.
According to the proof of \zcref{obs:restoring-L}, only the memberships in $L$ of vertices in $\{y' \in Y \mid \delta(y') \cap F \neq \emptyset\} \cup A$ are affected by $F$ and $A$.
Hence, it suffices to check for each $y'$ in the set if exactly one of $y' \in L$ or $y' \not\in L$ is deduced by $F$ and $A$ according to the proof of \zcref{obs:restoring-L}; otherwise, we remove $(y, F, A)$ from $\mathcal P$.

Now, we filter out the candidates to be removed in steps \zref{step:removing-invalid-candidates} and \zref{step:removing-trivially-no-candidates} at the same time.
Instead of checking the conditions for each candidate na\"ively, we (radix-)sort $\mathcal P$ by key $x_\ell(y)$ for $(y, F, A) \in \mathcal P$ with respect to $\sigma_X$ and check all the candidates in order.
We do this with a data structure that (1) maintains a virtual edge $e'$ with endpoints $x' \in X$ and $y' \in Y$ and the set of vertices $L'$ that are to the left of $y'$, (2) can change $x'$ to its immediately left vertex or its immediately right vertex in constant time, (3) can add a vertex to or remove a vertex from $L'$ in constant time, and (4) can answer the set of edges that would cross $e'$ in a 2-layer drawing of $G$ deduced by $\sigma_X$ and $L'$ in linear time with respect to the size of that set.
With this data structure, for each $p(y, L) = (y, F, A) \in \mathcal P$ we adjust $(x', y', L')$ to be $(x_\ell(y), y, L)$ and check if the set of edges crossing $e'$ is exactly $F$.
By additionally setting $(x', y', L')$ to be $(x_r(y), y, L)$, we are also able to check the candidate for step \zref{step:removing-trivially-no-candidates}, namely, the condition in \zcref{obs:condition-around-y-in-Y}.

This can be done in linear time, since the number of operations of the data structure to call is $O(|Y|)$ and the set size of operation (4) should not exceed a certain constant.
Note that for any two consecutive candidates in sorted $\mathcal P$ the difference in $L$ can be computed locally in constant time.
The latter argument on the set size of operation (4) can be observed by the following.
Assuming that there is a 2-layer fan-planar drawing of $G$ with $\sigma_X$, for every $y \in Y$ there exists $L$ with $p(y, L) \in \mathcal P$ such that $e_\ell(y)$ crosses at most $14$ edges.
Let us call such $p(y, L)$ a \emph{proper} candidate.
Since the maximum degree is constant, operations (2) and (3) of the data structure only increase the set size of (4) by a constant.
As we visit all the candidates going through within a constant ``edit distance'' from those proper candidates, the set size should always remain bounded by a constant.
Otherwise, namely if the set size exceeded while processing the candidates, we can immediately conclude that there is no 2-layer fan-planar drawing of $G$ with $\sigma_X$.

\paragraph{Steps \zref{step:compute-dp-predicates} and \zref{step:dp-over-P} in Linear Time.}
It is clear that step \zref{step:dp-over-P} can be done in linear time given the sets computed in step \zref{step:compute-dp-predicates}.
We show that, for each pair $p(y, L) = (y, F, A) \in \mathcal P$, we can compute $\mathcal P'$ satisfying the condition in constant time.
The next lemma almost suffices to show that.

\begin{lemma}\label{lem:few-previous-vertices}
    Let $\mathcal D$ be a 2-layer fan-planar drawing of $G$ and $y$ be a vertex in $Y$ such that $\chi_{\mathcal D}(e_\ell(y)) = F$ and the set of degree-1 neighbors of $x_\ell(y)$ to the left of $y$ is $A$.
    There are at most 69 candidates for the vertex $y'$ that is immediately to the left of $y$.
    Moreover, given $y$, $F$, and $A$, one can compute all of those candidates in constant time.
\end{lemma}

\begin{proof}
    We divide into three cases.

    Suppose $x_\ell(y) <_X x_r(y')$.
    Since $e_r(y') \in F$, there are at most $|F| \leq 14$ candidates for $y'$.

    Suppose $x_r(y') = x_\ell(y)$.
    Then, there are $|N(x_\ell(y)) \setminus \{y\}| \leq 13$ candidates for $y'$.

    Suppose $x_r(y') <_X x_\ell(y)$.
    By definition, every vertex between them is adjacent to neither $y$ nor $y'$.
    Hence, every edge incident to such a vertex crosses either $e_r(y')$ or $e_\ell(y)$.
    Due to $14$-planarity of $\mathcal D$, the sum of degrees of vertices between $x_r(y')$ and $x_\ell(y)$ is at most $28$.
    Now, let $x'' \in X$ be the leftmost vertex such that $x'' <_X x_\ell(y)$ and the sum of degrees of vertices between them is at most $28$. Since the maximum degree of $G$ is at most $14$, that sum is at most $28 + 14 = 42$. As $y'$ is adjacent to one of the vertices between $x''$ and $x_\ell(y)$, there are at most $42$ candidates for $y'$.

    Overall, there are at most $14 + 13 + 42 = 69$ candidates.
    Observe that the candidates can be enumerated in constant time accordingly.
\end{proof}

Observe that the fact that $y'$ is immediately to the left of $y$, together with $L$, determines the set $F'$ of edges crossing $e_\ell(y')$ and the set $A'$ of degree-1 neighbors of $x_\ell(y')$ that are to the left of $y'$.
Given $y'$, we can compute such $(y', F', A')$ in constant time locally.
By \zcref{lem:few-previous-vertices}, if $\mathrm{dp}(y, F, A)$ is true, there exists $y'$ such that $(y', F', A')$ computed above satisfies $\mathrm{dp}(y', F', A')$.
Hence, by accumulating such tuples we obtain $\mathcal P'$ with the condition and $|\mathcal P'| \leq 69$.
This completes the proof of \zcref{thm:one-sided}.

\section{Conclusions}

First, we gave a polynomial-time algorithm for testing 2-layer fan-planarity of a bipartite graph.
We remark that the degree bound in \zcref{lem:degree-bound} has not been optimized for simplicity, although it is directly linked to the running time.
\begin{figure}[ht]
    \centering
    \includegraphics[page=10]{figure.pdf}
    \caption{A reduced 2-layer fan-planar graph having a vertex of degree~7.}
    \label{fig:lb}
\end{figure}
However, \zcref{fig:lb} shows a certain limitation of our approach: there is a reduced graph having a vertex of degree~7, which is 2-layer fan-planar.
This indicates that our running time cannot be improved significantly by optimizing the bound of \zcref{lem:degree-bound}.
Hence, a natural next question is whether the problem admits a polynomial-time algorithm with practical running time.
For instance, linear-time algorithms are known for biconnected graphs~\cite{BinucciCDGKKMT17} and trees~\cite{BiedlC0MNR20}.

In contrast, for the one-sided case, we gave an implementable linear-time algorithm with an explicit bound on the table size for dynamic programming.
Again, many bounds throughout the proof, especially the constant in \zcref{lem:few-candidates-for-each-y}, are not finely optimized for simplicity.
We believe that they are indeed yet smaller and the actual running time of our algorithm is promising for real-world graphs.

\subsection*{AI Usage Statement}

Except for the proof of \zcref{lem:one-sided:degree-bound-linear-time}, the original ideas are due to the authors: GPT-6.1 Sol suggested the idea of applying the reduction rules for only one round to easily accomplish linear time.
GPT-6 Pro, GPT-6 Astra, and other prior models by OpenAI were used throughout checking and correcting the proofs and polishing and editing the text both originally written by the authors.

\bibliography{main}
\bibliographystyle{abbrvurl}

\end{document}